\documentclass[journal]{IEEEtran}

\usepackage{color}
\usepackage{graphicx}
\usepackage{amsfonts}
\usepackage{subfigure}
\usepackage{amsmath}
\usepackage{cite}

\usepackage{algorithmic}
\usepackage{amsthm}
\usepackage{epsfig}
\usepackage{epstopdf}
\usepackage{verbatim}
\usepackage{amsthm}
\usepackage{algorithmic}
\usepackage[linesnumbered,ruled,vlined]{algorithm2e}

\newtheorem{lemma}{Lemma}
\newtheorem*{definition}{Definition}
\newtheorem*{remark}{Remark}
\newtheorem{assumption}{Assumption}
\newtheorem{property}{Property}
\newtheorem{p}{P-}

\newcommand{\mat}{\boldsymbol}

\newcommand{\AND}{\mathbf{AND}}

\begin{document}

\title{Fair Dynamic Spectrum Management in Licensed Shared Access Systems}

\author{M.~Majid~Butt,~\IEEEmembership{Senior~Member,~IEEE,}~Irene~Macaluso,~Carlo~Galiotto~and~Nicola~Marchetti,~\IEEEmembership{Senior~Member,~IEEE}
\IEEEcompsocitemizethanks{\IEEEcompsocthanksitem M. Majid Butt was with the Irish Research Centre for Future Networks and Communications (CONNECT), Trinity College Dublin,
Dublin 2, Ireland. He is with Nokia Bell Labs, Paris France.
E-mail: majid.butt@ieee.org.}
\IEEEcompsocitemizethanks{\IEEEcompsocthanksitem I. Macaluso, C. Galiotto and N. Marchetti are with the Irish Research Centre for Future Networks and Communications (CONNECT), Trinity College Dublin,
Dublin 2, Ireland.
E-mail: \{macalusi, galiotc, nicola.marchetti\}@tcd.ie.}

\thanks{The material in this paper has been presented in part at PIMRC 2016 \cite{Majid_PIMRC:2016}.}
\thanks{The project ADEL acknowledges the financial support of the Seventh Framework Programme for Research of the
European Commission under grant number: 619647. We also acknowledge support from the Science Foundation Ireland
under grants No. 13/RC/2077 and No. 10/CE/i853.
}
}

\IEEEtitleabstractindextext{%
\begin{abstract}
Licensed Shared Access (LSA) is a spectrum sharing mechanism where bandwidth is shared between a primary network, called incumbent, and a secondary mobile network. In this work, we address dynamic spectrum management mechanisms for LSA systems. We propose a fair spectrum management algorithm for distributing incumbent's available spectrum among mobile networks. Then, we adapt the proposed algorithm to take mobile network operator's regulatory compliance aspect into account and penalize the misbehaving network operators in spectrum allocation. We extend our results to the scenario where more than one incumbent offer spectrum to the mobile operators in a service area and propose various protocols, which ensure long term fair spectrum allocation within the individual LSA networks. Finally, we numerically evaluate the performance of the proposed spectrum allocation algorithms and compare them using various performance metrics. For the single incumbent case, the numerical results show that the spectrum allocation is fair when the mobile operators follow the spectrum access regulations. We demonstrate the effect of our proposed penalty functions on spectrum allocation when the operators do not comply with the regulatory aspects. For the multi-incumbent scenario, the results show a trade-off between efficient spectrum allocation and flexibility in spectrum access for our proposed algorithms.
\end{abstract}

}

\maketitle

\IEEEdisplaynontitleabstractindextext

\IEEEpeerreviewmaketitle

\section{Introduction}
\label{sect:intro}

It is a time of unprecedented change, where traffic on telecommunication networks is growing exponentially, and many new services and applications are continuously emerging. The
advent of mobile internet has led to phenomenal growth of the mobile data traffic over the past few years. As the features of the envisioned technologies and services of the fifth generation (5G) and beyond mobile communication systems dictate, this trend is expected to continue for the years to come. Future 'bandwidth hungry' mobile broadband (MBB) communication services require additional spectrum. This goal can be reached with the following methods: (1) Clearing (a.k.a. refarming) spectrum and allocating it to MBB; (2) Sharing spectrum between
existing incumbents and mobile network operators (MNO); (3) Using millimeter wave technology.

In particular, spectrum sharing is seen by national regulators, in both Europe and USA, as a viable solution for allocating additional spectrum to MBB in a timely fashion, since
technologies that are capable to implement it already exist. There are two main recently proposed approaches to spectrum sharing in licensed bands:
\begin{itemize}
  \item Authorized Shared Access / Licensed Shared Access (ASA/LSA) \cite{Mueck:14, NSN-Qualcomm:11}. In ASA/LSA, MNOs can use (on an exclusive basis) the licensed spectrum owned
      by other incumbents when and where these incumbents are not using it. In this way, the incumbents are protected from harmful interference and the licensees benefit from the
      provision of predictable QoS. The band under consideration for LSA use is 2.3-2.4GHz in Europe.
  \item Citizens Broadband Radio Service (CBRS), which, in addition to highest-priority incumbents and high-priority licensed users, also allows low-priority unlicensed users to
      access the spectrum on a shared basis, as long as they do not interfere with higher priority users \cite{FCC:14,PCAST:12}. However, for the latter type of users there are
      no QoS guarantees. The band foreseen for CBRS deployment is the 3.5GHz band in the USA.
\end{itemize}
We focus on LSA networks in this work, though most of the proposed concepts can be generalized to other licensed spectrum sharing systems.

\subsection{Contributions}
In general, LSA system depends on various factors including static/dynamic geographic zones, length of long term agreement and potential increase in MNO capacity, etc. Most of the
works on resource allocation in LSA assume that MNOs get rights for exclusive use of the spectrum in a particular area (when the incumbent is not present) without any spectrum
sharing with other MNOs. However, if the MNO is not using the available incumbent spectrum because of lack of demand, the spectrum remains unused. Our approach is to extend the LSA
spectrum access framework to 'as required' model, such that the MNOs' spectrum demand is taken into account at the spectrum allocation instant to minimize/eliminate unused
incumbent spectrum.

To the best of our knowledge, this work is one of the first works for general operating scenarios in LSA networks, addressing both single and multi licensee and
incumbent cases, encompassing both resource allocation and policy violation and enforcement aspects, providing a quantitative study of dynamic spectrum management techniques; and establishing fundamental properties of LSA in a rigorous way.

We summarize the main contributions of the paper in the following:
\begin{itemize}
  \item We propose and evaluate the performance of a fair dynamic spectrum management algorithm for a single incumbent and multi-operator LSA network. The proposed algorithm works as a building block for the more challenging scenarios discussed in the rest of the work. The numerical results show
      that the spectrum allocation algorithm is fair for the network operators and does not favor any operator with more demand; but aims at minimizing the variance in mean allocated spectrum.
  \item We address LSA network operator regulation compliance framework. We impose penalty in terms of reduced allocated spectrum for the operators not complying with the LSA
      network use regulations. The proposed fair algorithm is generalized to account for the misbehaviour of the licensee operators. The effect of different penalty functions is demonstrated through numerical results and the results validate the underlying idea that the spectrum allocated to the misbehaving MNOs is not fair any more and decreases in proportion to their violation behaviour.
  \item Finally, we extend our results to the multi-incumbent multi-operator scenario. We discuss various possible operational opportunities for the future LSA systems, when more than
      one incumbent offer spectrum to more than one MNO in the same service area. Building on the proposed fair spectrum allocation algorithm for the single incumbent-multi-operator scenario and depending on the mode of operation, we propose different spectrum management protocols that coordinate radio resource allocation for the multi-incumbent multi-operator scenario and maintain independence and privacy of LSA network stakeholders (MNOs and incumbents). We characterize
      the fundamental properties of the proposed algorithms in a rigorous way by a set of theorems, and quantify the performance numerically.
\end{itemize}

It is important to notice that the MNOs may decide to offer different price for spectrum access, implying a priority ranking on the available spectrum. On the contrary, our goal is
to provide fairness in spectrum allocation to the MNOs regardless of their high demands. Under this model, the MNOs choose to pay the same $\$/$Hz cost and the proposed fair
spectrum access is more cost effective for the MNOs as compared to the exclusive spectrum or priority based LSA spectrum access. As every MNO pays the same price for spectrum
access, the utility function for the LSA system is to distribute the available bandwidth (BW) in a 'long and short term' fair manner.

\subsection{Related Work}
Spectrum sharing has been the subject of intensive investigation for the last decade or so \cite{Tehrani:2016}. The reader is referred to many surveys that have appeared in the last few years on Cognitive Radio Networks (CRNs), Opportunistic Spectrum Access (OSA) and Dynamic Spectrum Access (DSA) \cite{Masonta:2013, Tragos:2013, Paisana:2014, Tsiropoulos:2016}. These surveys focus on the applicability of a wide range of coordination protocols/methods, and spectrum access/allocation techniques, under various licensing regimes for different spectrum ranges. Most of the CRN approaches rely on sensing and opportunistic access to share spectrum resources among different users. On the other hand, LSA (and other licensed spectrum access approaches) does not include sensing and opportunistic access. As a result, most of the spectrum sharing approaches proposed in CRNs are not directly applicable to LSA. In the following, we focus on the works mainly addressing LSA and licensed sharing paradigm.

LSA regulatory framework promotes the idea of sharing on a licensed basis and is complementary to both extremes of having and not having the license to spectrum rights. In response
to the Radio Spectrum Policy Group (RSPG) report \cite{RSPG:2011}, several documents defining the LSA sharing framework and its standardization have been published by the
Electronic Communications Committee (ECC), European Commission (EC) and  European Telecommunications Standards Institute (ETSI) \cite{ECC:2014, ETSI:2013}.

Sparked by the above interest and support from regulators and standardization, several works have dealt with the important problem of resource allocation for LSA systems. We report here an account of the literature that we believe is mostly related to our work, and explain shortly how this work
differentiates as compared to the state of the art.

There are a number of works in literature that model resource allocation in LSA using auction-based mechanisms \cite{Wang:15,Wang:2015-dyspan, mcmenamy2016enhanced}  and game
theory \cite{Saadat:15,AlDaoud:2015}.
In \cite{Wang:15}, the authors present an LSA Auction (LSAA) mechanism to allocate incumbents' idle spectrum to the licensee access points from different operators, adopting a policy aiming for revenue and market regularity. In \cite{Wang:2015-dyspan}, the authors present a framework of joint auction with a mixed graph based on the LSA architecture, to coordinate the interference between different operators and accommodate more base stations, thus providing a fair competition environment for small firms competing for spectrum. In \cite{mcmenamy2016enhanced}, an auction-based approach to the LSA framework is proposed, with virtual network operators (VNOs) that share not only spectrum but also infrastructure in a cloud-based massive-MIMO system. The authors of \cite{Saadat:15} propose a two-tier evolutionary game for dynamic allocation of spectrum resources enabling coexistence of incumbents and LSA licensees, enabling fair decisions for spectrum allocation. In \cite{AlDaoud:2015}, the authors formulate the spectrum sharing problem as a non-cooperative iterated game of power control where service providers change their power levels to fix their long-term average rates at utility-maximizing values. In \cite{Hafeez:2015}, the authors consider a two-tier network where a small cell network offers offloading services to a macro network, while in return the small cells are rewarded with a number of licenses to operate in the spectrum owned by the macro network.

Other works focus on a range of different techniques to model resource allocation in the LSA framework. In \cite{Sadreddini:18}, an approach based on MNO decision policy is discussed for an LSA system that combines both pricing and rejection rules for the secondary users.
The authors of \cite{Borodakiy:14} study a one-cell LTE system using LSA, proposing a methodology to model the unreliable operation of an LSA frequency band, by employing a
multi-line queuing system with unreliable servers. In \cite{Luo:14}, a multi-carrier waveform based flexible inter-operator spectrum sharing concept is proposed, based on adapting waveforms with respect to the out-of-fragment radiation masks. In \cite{Lagunas:2015}, a resource allocation mechanism for a shared system is discussed where terrestrial networks act as an incumbent and satellite networks are the secondary operators, and joint power, bandwidth and carrier allocation schemes are proposed. The authors of \cite{Perez:2014} propose different methods to optimize cell selection and power allocation, by simulating an LTE network where an MNO is allowed to use the 2.3 GHz band as an ASA licensee. Wirth \emph{et al}. \cite{Wirth:2014} also consider antenna type and orientation, besides the types of resources addressed in \cite{Perez:2014}. In \cite{He:14}, a distributed antenna system architecture is considered and LSA is investigated on the downlink cell edge in network virtualization context.


To the best of our knowledge, none of the literature on LSA networks to date, has addressed both single and multi licensee and incumbent cases as in this work. We propose a framework incorporating both resource allocation and policy violation and enforcement aspects, and quantitatively study dynamic spectrum allocation aspects.

The rest of this paper is organized as follows. Section \ref{sect:system} introduces the system model and preliminaries used in this work. We discuss the proposed spectrum
management algorithm for the single incumbent multi operator case in Section \ref{sect:algorithm}. Section \ref{sect:enforcement} addresses regulation policy enforcement aspects of
LSA networks. Dynamic spectrum management for the multi
incumbent-multi operator case is dealt in detail in Section \ref{sect:multi-incumbent}. We present numerical results for the proposed algorithms in Section \ref{sec:results} and
summarize the main results in Section \ref{sect:conclusions}.
\section{LSA System Model}
\label{sect:system}
We briefly introduce the building blocks of LSA network first. Fig. \ref{fig:system} shows the system diagram for a typical LSA system with 2 incumbents and 3 licensee MNOs.
The LSA repository is responsible for managing the database record including information
such as, availability of incumbent spectrum, which MNO is using what part of the spectrum, and how long it is
permitted to use it in a specific service area \cite{Jush:2012}. It also contains the information about regularized/unregulated use of the assigned spectrum by the licensee operators.

The LSA band manager is responsible for controlling the spectrum access and provides information to the
Operation, Administration, and Maintenance (OA$\&$M) unit of each MNO about the allocated spectrum for its use. OA$\&$M has a control channel for communicating such information.
OA$\&$M is the MNO interface to the LSA system, and is responsible for the base station level allocation of the
assigned incumbent spectrum to the MNO.

\subsection{Modeling incumbent's activity over time}
In our LSA system, we consider $M$ incumbents that act as primary users for the LSA spectrum.
Each incumbent can reserve and use the spectrum for a given time frame, and this
information is stored in the LSA repository.
Our work focuses on dynamic spectrum access, by the MNOs, of
the spectrum made available by the incumbent.

\begin{figure}
\centering
\includegraphics[width=6cm]{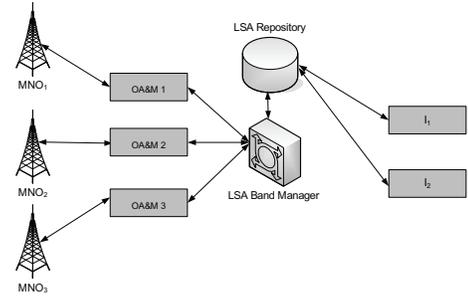}
\vspace{-0.3cm}
\caption{System diagram for an LSA system with 3 MNOs and 2 incumbents. Note that a single MNO can use more than one base station in a service area. An incumbent $m$ is denoted by $I_m$, while MNO $n$ is denoted by $MNO_n$.}
\label{fig:system}
\vspace{-0.4cm}
\end{figure}

Spectrum management is performed at different levels in LSA networks. To differentiate the different levels, we define some terms which will be required in the discussion on LSA spectrum access in the sequel.
\begin{definition}[Level 1 Algorithm]
The algorithm used to assign available spectrum from the incumbent at a specific time instant, in a specific
service area, to one or multiple licensee MNOs is called level 1 (L1) algorithm.
\end{definition}
When spectrum is available at the beginning of the frame, L1 spectrum allocation algorithm is run and the spectrum is distributed among the MNOs. The length $T_B$ of the frame is in the order of tens of minutes. We term the beginning of the frame when the spectrum is available as a 'spectrum allocation instant' in the rest of this work.
\begin{definition}[Level 2 Algorithm]
The algorithm used by the radio access network (RAN) of the individual licensee MNOs to (re)assign radio
channels to the base stations is termed as level 2 (L2) algorithm.
\end{definition}

The LSA resource allocation mechanism is limited to L1 algorithm. Once the spectrum is assigned to the individual MNOs through LSA L1 algorithm, they are free to use any L2 algorithm of their choice at network level, oblivious of L1 algorithm used. The focus of this paper is to study LSA spectrum management at level 1 while base station level resource allocation is out of the scope of this paper.

\section{Proposed Level 1 Algorithm}
\label{sect:algorithm}
The spectrum is distributed among the MNOs using L1 algorithm depending on the utility function. As already mentioned in Section \ref{sect:intro}, we choose fairness (in the mean sense) as utility function for the LSA spectrum management. We assume that the MNOs have agreed \emph{a-priori} on a fair use of shared resources such that every MNO pays the same price and agrees on receiving a fair share of the available LSA spectrum. This scenario can typically be useful for the MNOs which run services like some flavours of smart grid and vehicular communication applications, where the MNOs welcome additional LSA spectrum but the applications are not strictly time critical.


We propose an L1 algorithm for incumbent's spectrum access and measure system fairness by evaluating mean spectrum allocated to the MNOs over certain period of time. In a fair LSA system, the variance in mean allocated spectrum to the MNOs should be as small as possible.
We assume that $N$ MNOs are sharing the spectrum in a service area. For a recent spectrum allocation history window size $W$, we define a priority index $PI_{n}(t)$
for MNO $n\leq N$ in a time slot $t$ by,
\begin{eqnarray}\label{eqn:new_allocation}
PI_{n}(t)&=& \frac{\mbox {Rewarded BW to MNO $n$ in past}}{\mbox {Sum of rewarded BW by the incumbent in past}}\nonumber\\
&=&\frac{1}{W}\sum_{j=t-W}^{t-1} \frac{B_{n}^a(j)}{\sum_{n=1}^N B_n^a(j)}
\end{eqnarray}
where $B_n^a(t)$ denotes the spectrum allocated to MNO $n$ at instant $t$.
Based on $PI_{n}$, a possible approach to achieve fairness is to distribute the available spectrum at each spectrum allocation instant in a proportionally fair manner following the approach similar to weighted fair queuing \cite{Frascolla:2016}. Depending on its allocation history, every MNO gets a fair share of the available spectrum in every spectrum allocation instant, i.e.,
\begin{equation}
B_n^a(t)=B\frac{1-PI_n}{\sum_{n=1}^N(1-PI_n)}
\label{eqn:Fair1}
\end{equation}
where $B$ is the spectrum available from the incumbent.
If demand for some of the MNOs is less than the allocated spectrum $B_n^a(t)$, the surplus spectrum is divided among the remaining MNOs using (\ref{eqn:Fair1}) again.
However, this \emph{strictly fair} approach has the drawback that an MNO may get a small share of the spectrum at a particular spectrum allocation instant, which may not be enough for its use, e.g., smart grid network providers may prefer to get spectrum infrequently, but the spectrum slice should be large enough to meet their demands.

Our approach is to maintain fairness in the mean sense for spectrum allocation across the MNOs, but we aim at satisfying temporal spectrum requirement of the MNOs as much as possible. This approach is similar to proportionally fair scheduler (PFS) where a user is scheduled for transmission that \emph{maximizes} its rate normalized by throughput \cite{Tse}. The scheduled user is less prioritized in next time slot due to larger normalization factor. It is known that PFS provides fairness to the system users and each user gets $1/N$ of the available resources when the users have infinite backlog \cite{Borst:2005}. We design an L1 algorithm which provides a similar fairness measure to the users when their spectrum demand statistics are not identical and demand is large.

In an LSA system, an MNO $n$ computes its spectrum demand $B_n^d$ by evaluating rate requirements for its network.
The set of indices for all MNOs in a service area is defined by $\mathcal{S} =\{1,\dots, N\}$.
Based on available spectrum $B$ and $PI$ for the MNOs, the proposed fair L1 algorithm at each spectrum allocation instant is presented in Algorithm \ref{alg:L1}.

\begin{remark}
We do not differentiate the MNOs which do not require spectrum in composition of set $\mathcal{S}$. All the MNOs with no requirements are modeled by the MNOs with $B_n^{d}=0$. If they are head of queue (HOQ), they do not get any spectrum.
\end{remark}

\begin{algorithm}[!t]
\label{alg:L1}
\caption{The Proposed L1 Algorithm}
\KwIn{$\mat{PI},\mathcal{S},B,\mat{B^d}$;}
\tcc{Input PI vector $\mat{PI}=\{PI_1,\dots,PI_N\}$ and demand vector $\mat{B^d}=\{B_1^d,\dots,B_N^d\}$ for all MNOs $n\in \mathcal{S}$.}
Initialize $N$ dimensional allocated spectrum vector $\mat{B^a}=\{B_1^a,\dots, B_N^a\}$ with $\mat{B^a}=\{0,\dots, 0\}$ \;
\While{$\mathcal{S}\ne \emptyset$ $\AND$ $B>0$}{
Select the MNO $n^*$ with minimum \emph{PI} from the MNOs $n\in \mathcal{S}$\;
$B_{n^*}^a=\min(B_{n^*}^d,B)$\;
$\mathcal{S}=\mathcal{S}-\{n^*\}$\;
$B=B-B_{n^*}^a$\;
}
\KwOut{$\mat{B^a}$;}
\end{algorithm}
The spectrum allocation is not fair temporarily (in every time slot), but the MNO at HOQ with the smallest
$PI$ is the one with the least amount of spectrum allocated in the past and gets priority.
To evaluate the temporal fairness, we compute moving average of the spectrum allocated to the MNOs over a span of $W$ time slots. Moving average $\bar{B}_n^W(t)$ in a time slot $t$ is defined by average of the allocated spectrum $B_n^a$ to an MNO $n$ and given by,
\begin{equation}
\bar{B}_n^W(t)=\frac{\sum_{j=t-W+1}^t B_n^a(j)}{W}
\end{equation}
Moving average curve is a good measure of smoothness of the spectrum allocated to the MNOs. If moving average does not diverge much from the mean, it verifies that the algorithm is fair over shorter time horizons as well. We numerically show in Section \ref{sec:results} that the proposed algorithm exhibits very good convergence behaviour.

The proposed algorithm exhibits the following properties in terms of fair resource allocation:
\begin{p}
  $B_n^a\leq B_n^d$: No MNO $n$ gets spectrum more than its demand at any spectrum allocation instant.
  \end{p}
  \begin{p}The spectrum allocation aims at minimizing variance of the mean allocated spectrum to the MNOs, especially when their average demand is larger than the mean allocation (similar to fully backlog scenario in data scheduling).
  \end{p}
  \begin{p}
  Denoting mean allocated spectrum by an incumbent by $\bar{B}$, the mean spectrum allocated to the MNO $n$ with large average spectrum demand converges to $\bar{B}/N$, i.e., $\bar{B}_n^a\to\bar{B}/N$ and more demand does not guarantee more spectrum allocation as long as other MNOs make large average demands.
  \end{p}
  \begin{p}
  The share of any MNOs $\acute{n}\ne n$ is more than $\bar{B}/N$ in spectrum allocation only if at least one MNO $n$ makes average demand less than its mean share.
\end{p}
It is worth noting that large spectrum demand does not promise large spectrum allocation as long as demand from the other MNOs is large enough. The proposed L1 algorithm helps the MNO with smaller demand to ramp up its mean allocated spectrum up to $\bar{B}/N$ and then stops at this point.

The L1 algorithm is implemented using a band manager, while the LSA repository is responsible for maintaining data for the occupancy of the bands as a result of spectrum allocation decisions. The algorithm is computationally simple to implement as \emph{PI} for every MNO is updated in every time slot and low complexity L1 algorithm is run.

\section{LSA Spectrum Sharing Policy Enforcement}
\label{sect:enforcement}
When the spectrum is allocated to the licensee MNOs, they have to comply with the
regulations of LSA operation \cite{Galiotto:2018}. For example, an MNO $n$ can access the spectrum borrowed from an incumbent within
a certain service area, using a certain carrier frequency, and during the allocated time period. However, it is
possible that the licensee MNO violates the regulations by:
\begin{enumerate}
  \item Transmitting with more power than permitted, and causing interference out of the service area.
  \item Using a different carrier frequency than allocated.
  \item Using spectrum for more time than permitted.
\end{enumerate}
In this section, we propose a framework to penalize the misbehaving MNOs. The MNOs can violate the LSA spectrum
use regulation in any of the above mentioned domains, i.e., power, frequency or time. However, penalty framework can be introduced in one domain without any loss of generality. The amount of spectrum allocated to an MNO is the main utility for the licensee
operators. If they commit any of the above mentioned violations of LSA spectrum use regulations, it is sufficient to penalize them in future spectrum assignment.

The primary goal of the proposed L1 algorithm is to provide fairness in spectrum allocation. We address policy violation framework provided that the system is using fair L1 algorithm discussed in Section \ref{sect:algorithm}. It is apparent that fairness and policy compliance frameworks are contrasting requirements for a spectrum allocation algorithm and the policy compliance framework will have significant impact on fairness of the proposed algorithm; but this is exactly the purpose of the enforcement policy. This is one of the significant contributions of the work that we handle both contrasting requirements together. As stated before, all of the MNOs pay the same $\$/Hz$ cost and are happy to receive their fair share of spectrum. A common way of penalizing the misbehaving MNOs is to impose monetary charges by the regulator. We provide an interesting view to penalize the misbehaving MNOs by monitoring their behaviour over time and reducing their mean share of spectrum, while the MNOs complying with the regulations get benefit in terms of extra spectrum allocated.

To integrate the regulatory framework in the proposed L1 algorithm, let us define a Penalty Index ($PEI$) for an MNO $n$ by,
\begin{eqnarray}
&&PEI_{n}(t)=\nonumber\\
&&\lim_{W\to \infty}\frac{\sum_{j=t-W}^{t-1}I\mbox{(Spectrum rule violated at instant $j$)}}{\sum_{j=t-W}^{t-1}I\mbox{(Spectrum assigned at instant $j$)}}\nonumber\\
&&=\frac{N_{v,n}}{N_{a,n}}
\label{eqn:PEI}
\end{eqnarray}
where $I(.)$ denotes the indicator function, which is $1$ if the argument is true. $N_{v,n}$ and $N_{a,n}$
denote the number of times spectrum rule was violated and the number of times spectrum was
assigned to an MNO $n$, respectively. The area of system level mechanisms to identify
misbehaving MNOs in LSA is well researched upon, e.g., references \cite{Zhang:2012,Jin:TMC} propose different mechanisms to identify the misbehaving users in the shared systems. In contrast to study of misbehavior detection, we confine ourselves to the study of appropriate penalty mechanisms for the MNOs not complying with the LSA enforcement regulations.

The proposed LSA L1 algorithm has fairness utility to be maximized, while the misbehaving MNOs can be penalized
when performing spectrum allocation. We merge $PEI$ with $PI$ to have a new metric selection index ($SI$), which
encompasses the spectrum rule violation framework as well.

A cumulative selection index is thus defined as,
\begin{equation}
\label{eqn:enforcemnet}
SI_{n}= \begin{cases}
          \omega PI_{n}+(1-\omega)f(PEI_n), & \mbox{if }  f(PEI_n)\leq\kappa_{PEI} \\
          \infty, & \mbox{otherwise}.
        \end{cases}
\end{equation}
where $0\leq \omega\leq 1$ is a variable which we use to change the penalty weight and
$f(PEI_n)$ is a general penalty function whose values vary between 0 and 1. $\kappa_{PEI}$ is a violation threshold on the value of $f(PEI_n)$. The $SI$ is evaluated only if the value of $f(PEI_n)$ for an MNO $n$ with $PEI_n$ is less than the violation threshold. $SI$ is set to infinity if $f(PEI_n)>\kappa_{PEI_n}$ to prevent the MNO $n$ from taking part in the spectrum access for a certain 'cooling off' period. The details of 'cooling off' period and $\kappa_{PEI}$ are left to the LSA system operators. Here we assume that $\kappa_{PEI}$ is high enough and all the competing MNOs comply with the regulations to the extent that no one is thrown out of contention.

The proposed L1 algorithm is implemented using
$SI_n$ index instead of $PI_{n}$ in the algorithm presented in Section \ref{sect:algorithm}.
It is worth noticing that the definition in (\ref{eqn:enforcemnet}) represents a general class of index/score definition where different weights ($\omega$ in our case) are assigned to
different utilities. It is left to the individual network operators to set 'scores' based on network key performance indicators of their preference \cite{Bonald:2004}. Therefore, the proposed framework in this paper is equally applicable to any general definition of $SI$.

A careful look at the use of $SI$ calculation in (\ref{eqn:enforcemnet}) for L1 algorithm unfolds a negative
feedback dilemma. A ${PEI}> 0$ reduces the probability of spectrum access for an MNO at the current allocation instant. However, by not allocating spectrum in the current spectrum allocation instant, we enhance the MNO's priority for the next spectrum allocation instant, because the proposed L1 algorithm attempts to implement fairness by decreasing $PI$ if the spectrum share of an MNO reduces. The $PI$ calculation has no mechanism to know whether the spectrum reduction is due
to 'deliberate' penalty imposed by the LSA regulator. Thus, a penalty imposed on a misbehaving MNO will not
hurt the MNO in the long run. This is what we mean by negative feedback.

To overcome this problem, we propose a slight modification in the originally proposed L1
algorithm. First, we do the spectrum allocation in each iteration of the algorithm based on $SI$ in (\ref{eqn:enforcemnet}) as before and inform the MNOs. At the same time, we perform spectrum allocation decisions based on $PI$ (solely) without making actual spectrum assignment. The 'fictitious' spectrum allocation decisions made on the basis of $PI$ are stored in a separate database. At the next spectrum allocation instant, the algorithm computes the value of $PI$ in (\ref{eqn:new_allocation}) on the basis of the 'fictitious' spectrum assignment from the database. This value is further used in (\ref{eqn:enforcemnet}) to compute $SI$. In this way, $PI$ computation is completely oblivious of the negative penalty due to regulatory violation and avoids negative feedback phenomenon.
\subsection{Penalty Functions}
\label{sect:penalty}
In this section, we propose two penalty functions which have specific characteristics:
\begin{enumerate}
  \item \textbf{Linear function:} In this case, $f(PEI_n)=PEI_n$ and all the MNOs are penalized on a linear scale
      depending on their regulatory violation statistics in (\ref{eqn:PEI}).
  \item \textbf{Power function:} In this case, $f(PEI_n)=(PEI_n)^{c}$ where $c$ is a positive constant. This
      function grows slowly in the beginning and much faster as $PEI$ increases. It is left to the individual
      LSA regulators to decide how to construct the power function. The rationale behind the
      power penalty function is to penalize the offenders mildly in the beginning and increase the penalty at a faster rate as the offense increases.
\end{enumerate}
The larger values of $c$ make the function growth rate faster. We discuss the effect of choosing weight $\omega$ and parameter $c$ through numerical evaluation in Section \ref{sec:results}.


\section{Multi-MNO--Multi-Incumbent Scenario}
\label{sect:multi-incumbent}
We have considered L1 algorithm for the scenario with a single incumbent and multiple MNOs in the previous section. We extend our discussion to more complex LSA system where we have $M$ incumbents in a service area, which offer spectrum to $N$ MNOs at the same time. This system settings open new opportunities for the operators and incumbents to share the network resources and are of great interest for future networks' operation. A straight forward centralized solution is to sum the available spectrum from all the incumbents and distribute it among the MNOs using the proposed L1 algorithm. However, it is not necessary that all of the incumbents and the MNOs share one LSA agreement. An MNO $n\in N$ can have separate LSA agreements with multiple incumbents. Therefore, we define a multi-incumbent spectrum management protocol, called L0 protocol hereafter. L0 protocol allocates spectrum to the MNOs by using L1 algorithms specific to the incumbent without sharing spectrum information with the MNOs that do not share the LSA agreement. Similar spectrum coordination protocols in the context of  inter--operator network sharing have been proposed in literature, e.g., \cite{Singh:commag15}.

Let us denote the set of incumbent indices available with surplus bandwidth in a service area by $\mathcal{I} =\{1,\dots, M \}$. Let us define the notion of \emph{LSA coalition} for an incumbent $m\in \mathcal{I}$ and MNOs $n\in \mathcal{S}$ by a set of one incumbent and one or multiple MNOs, which share one LSA agreement; and define by,
\begin{equation}\label{eqn:consor}
  \Omega_m =\big\{\{m,C\}:m\in I,C\subseteq \mathcal{S},C\ne \emptyset \big\}.
\end{equation}
Any L0 resource management protocol must meet the following requirements for the incumbent and the MNOs:
\begin{itemize}
  \item Every incumbent $m$ makes fair spectrum allocation for the MNOs in its LSA agreement without taking care of the MNOs' individual interest.
  \item At a spectrum allocation instant, the spectrum allocation decision of an incumbent is local, and not influenced by the spectrum allocation decisions of the other incumbents, which maintains independence of each LSA network.
  \item Every MNO aims at maximizing its short term reward (spectrum) at each spectrum allocation instant by considering the spectrum offers from various incumbents without taking care of system fairness.
\end{itemize}
These requirements imply that the inter operability protocol we propose is a facilitation mechanism for multi-MNO-multi-incumbent scenario and maintains the individual independence of every coalition, while allowing each MNO to maximize its spectrum allocation by jointly considering all the spectrum offers.

Next, we propose spectrum management mechanisms with the aim that the LSA network would like to distribute the total available spectrum efficiently (by minimizing unused spectrum) and fairly.
\begin{assumption}
We assume that all the MNOs and the incumbents agree on achieving the common goal (utility) of system fairness and efficient use of the shared spectrum resources following the requirements of LSA and mutual independence of coalitions.
\end{assumption}

\begin{assumption}
We assume that all the incumbents make the spectrum available at the same instant, i.e., the spectrum allocation instants are synchronized. It is worth noting that this is the only information the incumbents share with each other through a centralized node.
\end{assumption}
Following the requirements of the LSA operation, the proposed algorithms hold the following properties:
\begin{property}
The incumbent $m$ shares available spectrum information and its allocation decisions with only the MNOs, which share coalition $\Omega_m$.
\end{property}
This property is in line with the LSA policy. There is a separate licensing agreement between every licensee MNO and the incumbent. One MNO may have licensing agreements with more than one incumbent, but the incumbents operate independently and may not like to share any information with any other incumbent. Therefore, an MNO $n$ can have a centralized local knowledge of spectrum allocation (L1 algorithm) decisions of different coalitions with $n\in \Omega_m$, but an incumbent $m$ has no access to the spectrum allocation decisions of any other incumbent $\acute{m}\ne m$. The L0 protocol is a decentralized algorithm, which manages spectrum allocation process for all the MNOs and incumbents.
\begin{property}
Every incumbent $m$ runs L1 algorithm independently based on its own spectrum allocation history for the MNOs in $\Omega_m$.
\end{property}
This is in line with the requirements above. The spectrum allocation history for every MNO $n$ is specific to its coalition $\Omega_m$. Therefore, different incumbents $m$ make different spectrum allocation decisions when we run L1 algorithm independently.

Without loss of generality, we assume that all MNOs in set $\mathcal{S}$ have independent LSA agreements with all $M$ incumbents\footnote{It is worth noting, if an MNO is not in LSA agreement, it can be modeled by receiving an offer with zero bandwidth from the incumbent.}. Therefore, the number of MNOs in a single coalition and the number of coalitions in a service area is $|\mathcal{S}|$ and $|\mathcal{I}|$, respectively.
Based on different operating conditions and flexibility in spectrum access, we propose three possible scenarios for LSA operation.
\subsection{One Incumbent per MNO}
We allow spectrum access from \emph{at most} one incumbent for a single MNO at a single spectrum allocation instant and call it One-to-One scenario (OOS). The rationale is that using various carriers at two different frequencies for a single MNO may not be a preferred solution due to hardware, signalling issues and complications in carrier aggregation of non-contiguous frequency bands \cite{Bai:2012,Park:2013}.

When spectrum of the incumbent $m \in \mathcal{I}$ is assigned to an MNO $n \in \mathcal{S}$, we denote this assignment by $n \to m$.

The proposed L0 algorithm/protocol for this scenario works as follows:
\begin{enumerate}
  \item All $n\in \mathcal{S}$ MNOs provide their spectrum requirements $B_n^d$ to all $m\in \mathcal{I}$ incumbents with available spectrum $B^m$ through repository.
  \item L1 algorithm server for every coalition $\Omega_m$ individually runs L1 algorithm over the MNOs in $\Omega_m$ and $n\in \mathcal{S}$, and the available bandwidth $B^m$ from the incumbent $m$.
  \item Every MNO $n\in{\mathcal{S}}$ receives an offer $A_n^m\in \big(0,\min(B_n^d,B^m)\big)$ from every incumbent $m\in\mathcal{I}$.
  \item Every MNO $n\in{\mathcal{S}}$ chooses the best offer $\hat{A}_n=\max(A_n^1,\dots, A_n^M)$. The selected MNO $n^*$ is the one with the $\max(\hat{A}_1,\dots, \hat{A}_N)$.\footnote{The ties are broken by making a random selection.} It is assigned its selected incumbent $m^*$ and the corresponding spectrum from incumbent $m^*$ is $B_{n^*,m^*}^a=\hat{A}_{n}$. After assignment $n^*\to m^*$, $\mathcal{S}=\mathcal{S}-\{n^*\}$ for the next round.
  \item For the selected incumbent $m^*$, the band manager updates its available spectrum by $B^{m^*}=B^{m^*}-B_{n^*,m^*}^a$. If $B^{m^*}=0$, the incumbent $m^*$ opts out of the next round and is removed from the set $\mathcal{I}$ such that $\mathcal{I}=\mathcal{I}-\{m^*\}$.
  \item All the MNOs $n\in \mathcal{S}$ and the incumbents $m\in \mathcal{I}$ take part in the next round of the algorithm and go back to step 1.
  \item The process terminates when either $\mathcal{I}=\emptyset$ or $\mathcal{S}=\emptyset$.
\end{enumerate}
The offer $A_n^m\in\big(0,\min(B_n^d,B^m)\big)$ (in step 3) above shows the possible outcomes: the MNO $n$ is not offered any spectrum after running L1 algorithm for $\Omega_m$ or it is offered the minimum of its demand and the maximum available spectrum from incumbent $m$. Note that we used the term 'round' for one run of the algorithm from step 1 to 6.

As allocation history is different for every coalition $\Omega_m$, different MNOs can be proposed to be assigned to different incumbents in one algorithm round. However, it is necessary to assign only one $n\to m$ pair in one round (as in step 4) by the following Lemma 1.
\begin{lemma} For OOS algorithm, assigning more than one MNO-incumbent pair $n\to m$ in one round of the algorithm is not optimal for maximization of spectrum allocation for the competing MNOs.
\end{lemma}
\begin{proof}
Please see Appendix \ref{sect:proof_lem1} for the proof.
\end{proof}

In every round of the algorithm, one MNO is allocated an incumbent spectrum and removed from the allocation process. If $N\leq M$, the number of rounds are $N$. If $N>M$, the minimum number of rounds is $M$ while the maximum number of rounds is $N$ as one incumbent is allowed to be assigned to multiple MNOs if spectrum is available after first assignment.
\subsection{One to One Connection}
\label{sect:OOC}
This scenario is very similar to OOS in operation. In contrast to OOS, once an MNO is assigned an incumbents's spectrum, \emph{both} incumbent and MNO do not take further part in the spectrum allocation process and this scenario is termed as one-to-one connection (OOC). Note that OOS can allocate a single incumbents's spectrum to multiple MNOs. In this mode of LSA operation, the wasted un-allocated incumbent spectrum is relatively high, which we quantify in Section \ref{sec:results}. Once some part of incumbent spectrum is allocated to the MNO, the remaining spectrum cannot be used by other MNOs even if they need it.

For the L0 protocol for the OOC, first 4 steps are exactly the same as in OOS. The other steps are modified as follows:
\begin{enumerate}
\setcounter{enumi}{4}
  \item For the selected incumbent $m^*$, $B^{m^*}=0$ after assignment $n^*\to m^*$. The incumbent $m^*$ opts out of the algorithm. It is removed from the set $\mathcal{I}$ such that $\mathcal{I}=\mathcal{I}-\{m^*\}$.
  \item If $\mathcal{I}\ne\emptyset$ and $\mathcal{S}\ne\emptyset$, all the MNOs $n\in \mathcal{S}$ and the incumbents $m\in \mathcal{I}$ take part in the next round of the algorithm and go back to step 1.
  \item The process terminates when either $\mathcal{I}=\emptyset$ or $\mathcal{S}=\emptyset$.
\end{enumerate}
Lemma 1 is valid for the proposed L0 algorithm for the OOC scenario as well.\footnote{However, the proof is slightly different, but we avoid the proof to avoid repetition of results.}

In OOC, the number of rounds are $\min(N,M)$.
In Section \ref{sec:results}, we quantify the wasted spectrum of the incumbent for the OOC algorithm numerically.

\subsection{More than one incumbent per MNO}
In this scenario, we allow one MNO to utilize the spectrum available from multiple incumbents at a single time instant and call it multiple-connection scenario (MCS). This is the least restrictive scenario for the LSA multi--incumbent--multi-operator case.

The first 3 steps are the same as in the L0 protocol for the OOS scenario. The other steps for the proposed algorithm for this scenario work as follows:
\begin{enumerate}
\setcounter{enumi}{3}
  \item Let us define a temporary variable $\acute{B}_n^d$ at the beginning of every round which is initialed to $\acute{B}_n^d=B_n^d$. Every MNO $n$ chooses the best offer,
  \begin{equation}
  \hat{A}_n=\max(A_n^1,\dots, A_n^M)
  \label{eqn:alg1_maxmization}
  \end{equation}
  The MNO accepts the offer $\hat{A}_n$ if $\hat{A}_n>0$. The MNO updates its demand through temporary variable  ${\acute{B}}_{n}^d=\acute{B}_{n}^d-B_{n,m^*}^a$. If ${\acute{B}}_{n}^d>0$ after update, it remains in $\mathcal{S}$. The selected incumbent opts out by $\mathcal{I}=\mathcal{I}-\{m^*\}$, as ${\acute{B}}_{n}^d>0$ implies that the selected incumbent's available spectrum $B^{m^*}=0$ after assignment to MNO $n$. The MNO $n$ accepts the next best offer by applying (\ref{eqn:alg1_maxmization}) again on the incumbents in $\mathcal{I}$, but the allocated spectrum from an offer $A_n^m$ (if any) is reduced to $\min({B}^m,\acute{B}_n^d)$ (as compared to original offer $A_n^m=\min(B^m,B_n^d)$). $\acute{B}_n^d$ is updated again. The MNO $n$ stops iterations when either $\hat{A}_n=0$ or the demand variable $\acute{B}_{n}^d=0$. Note that every MNO $n$ simultaneously and independently performs continuous iterations in step 4 for all the offers $A_n^m>0$.
  \item At the end of each round, every MNO $n$ updates its allocated spectrum by $B_n^a=\sum_{m=1}^M B_{n,m}^a$.
  \item  For every $n\in \mathcal{S}$, the MNO takes part in the next algorithm round and $B_n^d=\acute{B}_{n}^d$.
  \item All the incumbents update their residual available spectrum by $B^m=B^m-\sum_{n=1}^N B_{n,m}^a$. The incumbents with $B^{m}=0$ opt out of the next round by leaving the set $\mathcal{I}$. Go back to step 2.
  \item The process terminates when either $\mathcal{I}=\emptyset$ or $\mathcal{S}=\emptyset$.
\end{enumerate}
The algorithm is characterized by the following lemmas.
\begin{lemma}
For MCS algorithm, allowing spectrum assignment (if offered) from more than one incumbent per MNO in one round (step 4) is more efficient in terms of speeding up algorithm termination, without any loss of potential spectrum allocation for a particular MNO.
\label{lem:MCS}
\end{lemma}
\begin{proof}
Please see Appendix \ref{sect:proof_lem2} for the proof.
\end{proof}
\begin{lemma}
Allowing all the distinct MNOs with positive offers ($A_n^m>0$) to get assigned to their respective incumbents in one round (step 4) is more efficient for speeding up algorithms's termination; without any loss of potential spectrum allocation for a particular MNO.
\label{lem:MCS2}
\end{lemma}
\begin{proof}
Please see Appendix \ref{sect:proof_lem3} for the proof.
\end{proof}

In MCS, both MNOs and the incumbents join the next round if the spectrum demand is not met and available spectrum is not finished, respectively. If $N<M$, the minimum number of rounds is $N$ and the maximum $M$. If $N>M$, the minimum number of rounds is $M$ while the maximum number of rounds is $N$. Thus, the number of rounds vary from $\min(N,M)$ to $\max(N,M)$.

\begin{figure*}[!t]
 \centering
 \subfigure[Fair L1 Algorithm]
 {\includegraphics[width=2.3in]{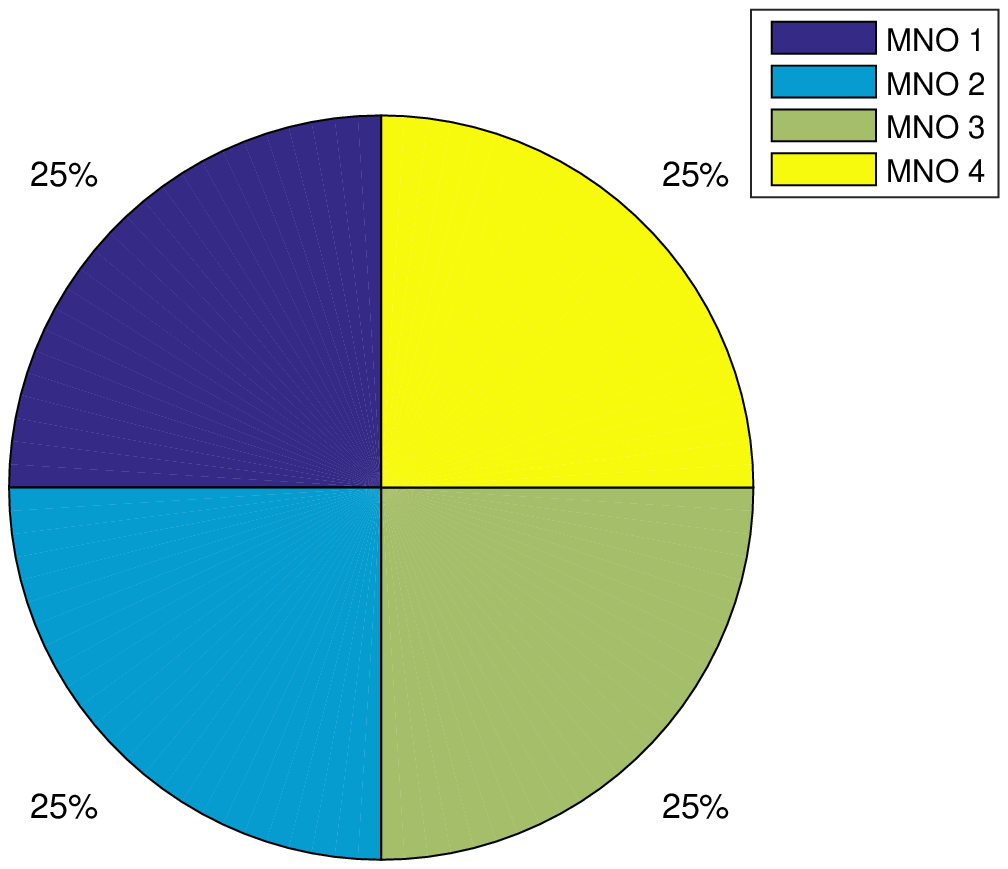}
 \label{fig:Fair-allocation}}
\subfigure[Round Robin L1 Algorithm]
 {\includegraphics[width=2.3in]{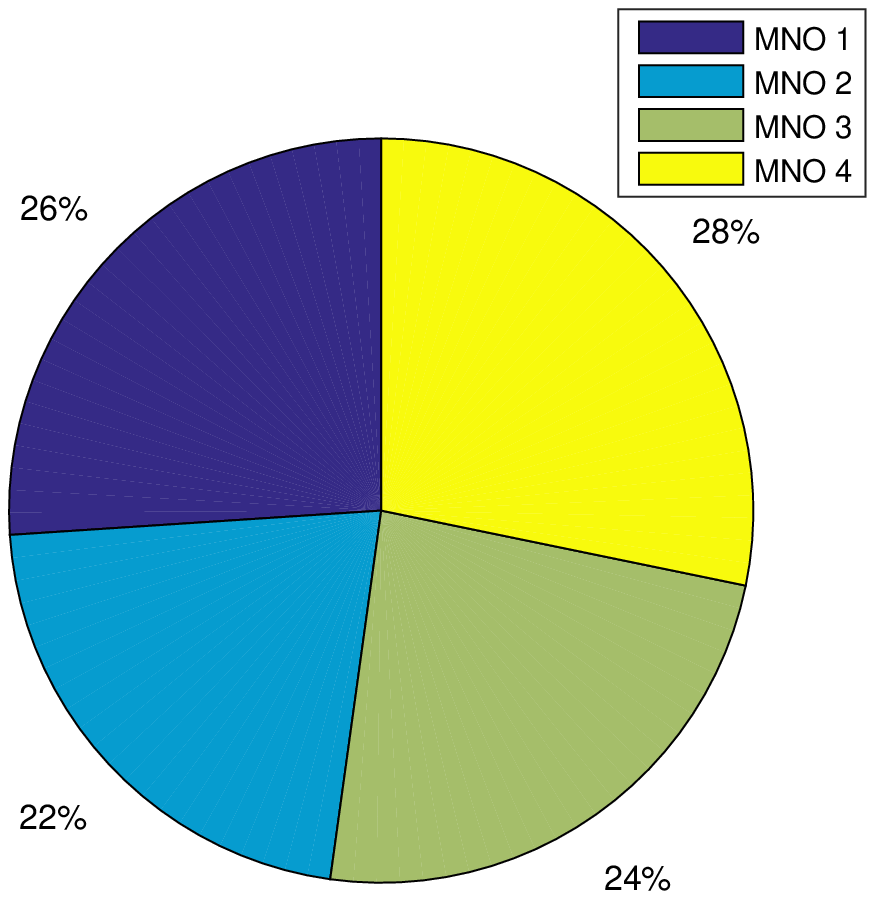}
 \label{fig:RR_allocation}}
 \vspace{-0.3cm}
  \caption{Mean spectrum allocation for 4 licensee MNOs in percentage.}
  \label{fig:allocation}
  \vspace{-0.4cm}
\end{figure*}

\section{Performance Evaluation}
\label{sec:results}
We use Monte Carlo simulations to evaluate the performance of the proposed algorithms. The
window size $W$ for computing $PI$ is set to 20 to ensure temporal fairness. As $PI$ computation for each MNO requires spectrum allocation in
last $W$ instants, we initialize simulations by having $W$ time slots with zero spectrum allocation and
random $PI$ (between 0 and 1) values for every MNO. In the
simulations, we consider $N=4$ and incumbent spectrum $B$ is normalized to 100 units without loss of
generality. We simulate $10^4$ spectrum allocation instants to compute mean
spectrum allocation for each MNO.

Let us define the mean allocated spectrum to an MNO (normalized by the offered incumbent spectrum) by,
\begin{equation}
\bar{B}_n^a=\frac{1}{T}\sum_{t=1}^T \frac{B_n^a(t)}{B(t)}
\end{equation}
where $T$ is the number of spectrum allocation instants.

\subsection{Single Incumbent Scenario Evaluation}
Fig.~\ref{fig:allocation} shows the mean spectrum allocation to 4 MNOs for the proposed fair L1 algorithm and Round Robin (RR) L1 algorithm. The RR algorithm provides spectrum access to all competing MNOs on RR basis. On its turn, the MNO is assigned spectrum according to its demand if possible. If the incumbent still has some spectrum available, next MNO gets its turn. If the demand cannot be met due to scarcity of incumbent spectrum, the MNO gets whatever the spectrum is available on its turn. On next spectrum allocation instant, spectrum allocation starts with the next MNO's turn.

\begin{figure*}
\centering
 \subfigure[MNO 1]
  	{\includegraphics[width=2.3in]{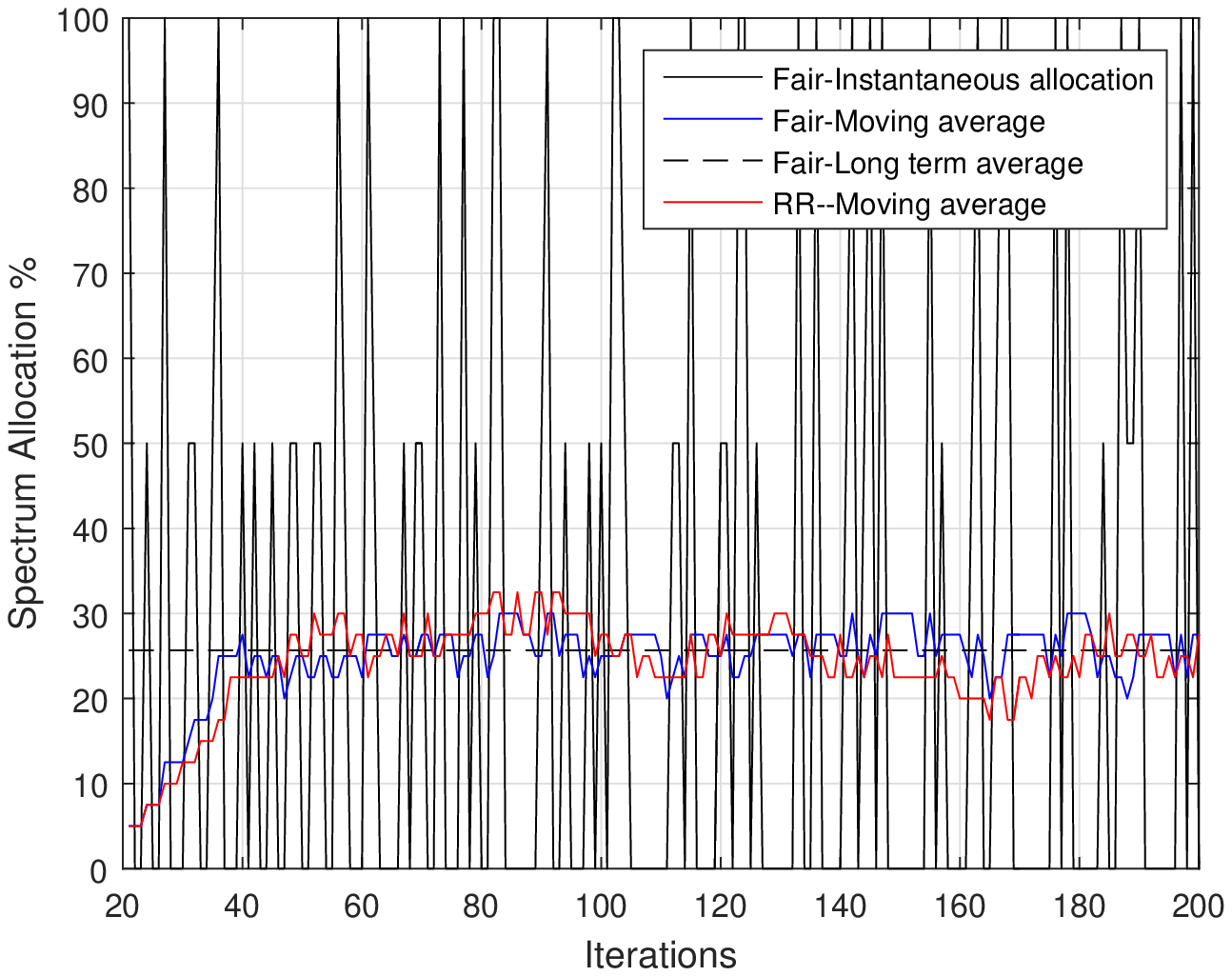}
  \label{fig:n1}}
  \subfigure[MNO 4]
  	{\includegraphics[width=2.3in]{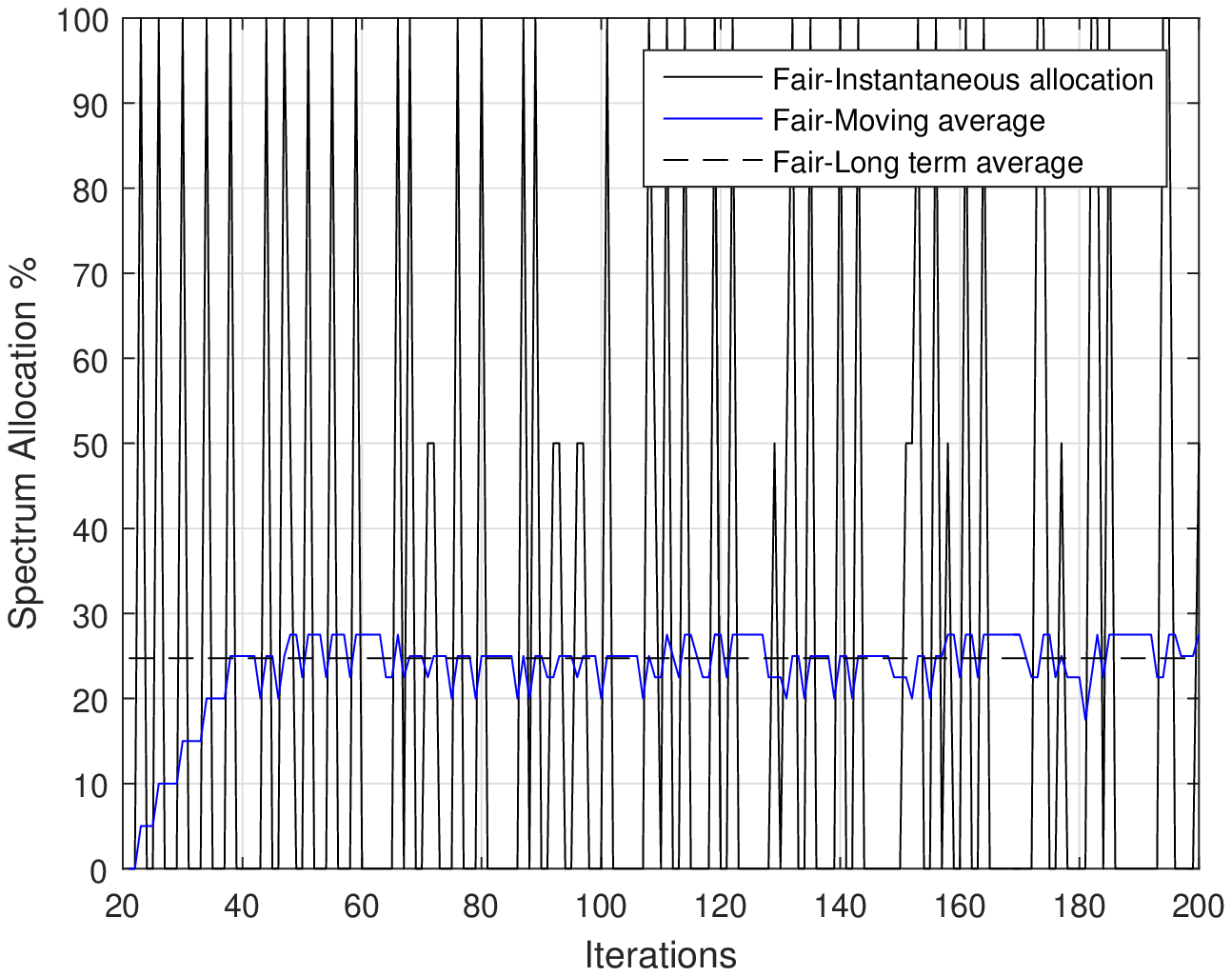}
  \label{fig:n4}}
   \subfigure[Performance of Fair Algorithm in \cite{Frascolla:2016}.]
  	{\includegraphics[width=2.3in]{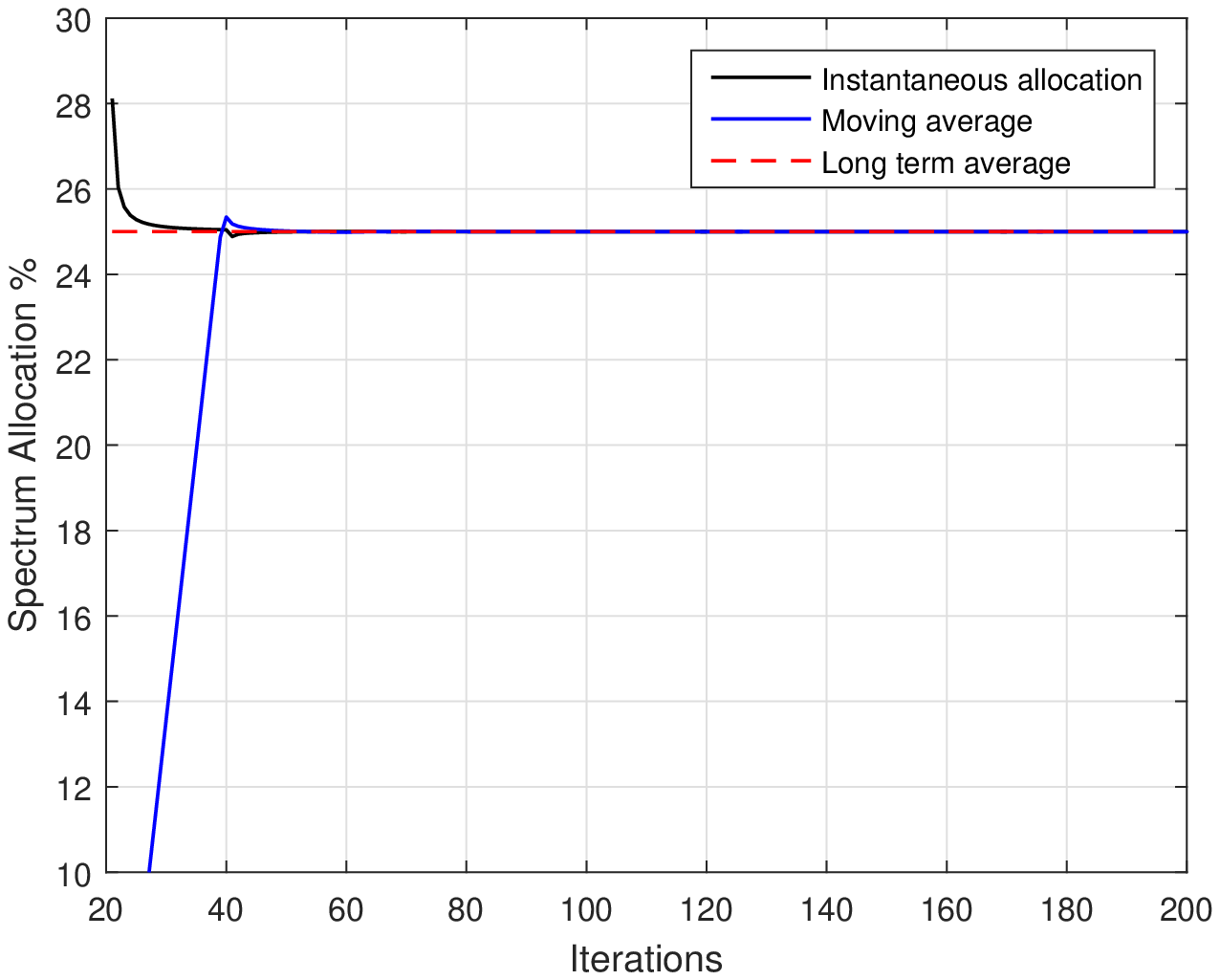}
  \label{fig:SFA}}
   \caption{Performance evaluation for short term spectrum allocation for MNO 1 and MNO 4. The 21-200 spectrum allocation
   instants are plotted (where first 20 instants are initialized with a random $PI$ for
   every MNO.)}
	\label{fig:fairness}
\vspace{-0.3cm}
\end{figure*}

At each spectrum allocation instant, MNO 1, 2 and 3 choose the spectrum demand randomly out of a vector of values [50,100] with uniform probability, while the spectrum demand for MNO 4 is fixed to 100. It is clear in Fig.~\ref{fig:Fair-allocation} that our proposed algorithm is fair in spectrum allocation and uniformly distributes spectrum among the MNOs regardless of higher demand from the MNO 4. On the other side, it is shown in Fig.~\ref{fig:RR_allocation} that RR algorithm is not fair and allocates more spectrum to MNO 4 with more demand. We observe that the demand for MNO 4 affects the spectrum allocation of the other MNOs for RR algorithm and different demands result in different mean $\%$ allocations for the other MNOs.

Fig.~\ref{fig:fairness} shows the instantaneous spectrum allocation statistics for the proposed fair L1 algorithm and other competing algorithms. As MNOs 1, 2 and 3 have symmetric spectrum demand distribution and allocation statistics, we plot statistics for MNO 1 only in Fig. \ref{fig:n1}, while statistics for MNO 4 are plotted in Fig. \ref{fig:n4}. The instantaneous allocation for the MNOs varies between either zero and its demand. As clear from
Fig.~\ref{fig:fairness}, when the MNO is allocated full spectrum, it has relatively low chance of accessing the
spectrum in the next few allocation instants. Similarly, a long sequence of zero allocation is usually
followed by large spectrum allocation. This justifies the algorithm's aim to achieve fairness in spectrum allocation for the MNOs. Although, MNO 1 and 4 are allocated the same mean spectrum in spite of more mean demand of MNO 4, the short term allocation pattern varies. MNO 4 gets full spectrum of 100 units for more instants as compared to MNO 1, but $PI$ makes sure that short term fairness is achieved by reducing the number of instances when MNO 4 is allocated spectrum.

To study the short term behaviour of the proposed algorithm, we plot the moving average of the spectrum allocated to the MNOs.
It is evident form Fig.~\ref{fig:n1} and Fig.~\ref{fig:n4} that the moving average of the allocated spectrum for MNO 1 and MNO 4 converges to its mean after very few iterations and diverges marginally from the mean afterwards. For comparison purpose, we also plot the moving average for MNO 1 for RR L1 algorithm in Fig. \ref{fig:n1}, which shows more variation around the mean as compared to our fair algorithm due to inability of the RR algorithm to adapt to the spectrum allocation history.

Fig. \ref{fig:SFA} shows the performance of fair algorithm for LSA in \cite{Frascolla:2016} where incumbent's spectrum is distributed among the MNOs using weight fair queuing principle and (\ref{eqn:Fair1}). The algorithm is fair in the mean sense as well, while moving average converges to mean after a few iterations. This is attributed to the reason that $PI$ of all the MNOs converge to the same value after the initial period and the spectrum is equally distributed among the MNOs in each allocation instant afterwards. This is strictly fair, however, may not be a desirable option for the MNOs who prefer large share of spectrum.

\begin{figure}[!t]
 \centering
 \subfigure[Linear penalty function]
 {\includegraphics[width=2.3in]{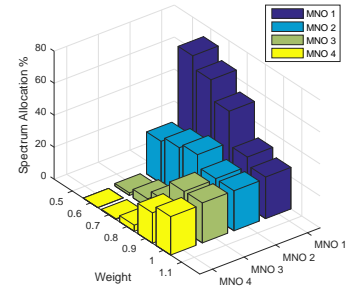}
 \label{fig:linear_enforcemnt}}
\subfigure[Power penalty function]
 {\includegraphics[width=2.3in]{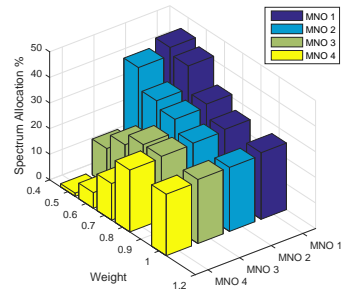}
 \label{fig:exp_enforcement}}
 \vspace{-0.3cm}
  \caption{Mean spectrum share for 4 licensee MNOs for different re-enforcement penalty functions.}
  \label{fig:enforcement}
  \vspace{-0.5cm}
\end{figure}

\subsection{Policy Enforcement Penalty Evaluation}

\begin{figure*}[!t]
 \centering
 \subfigure[OOS]
 {\includegraphics[width=2.3in,trim={0 2cm 0 0cm},clip]{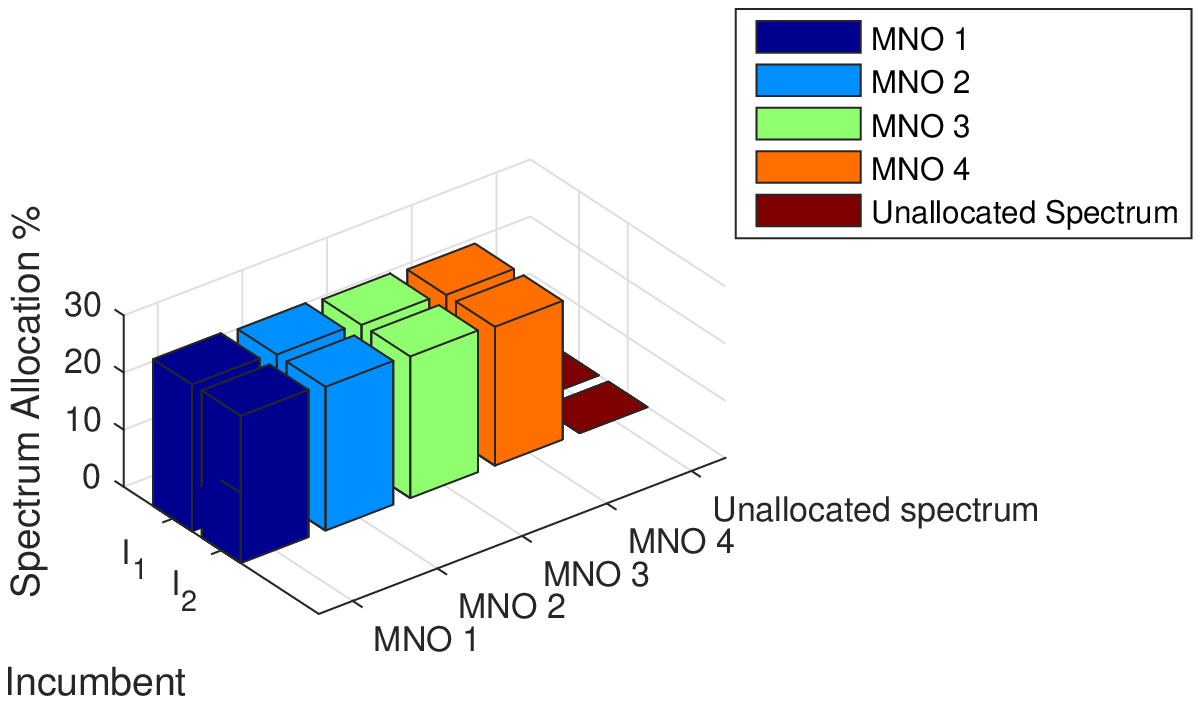}
 \label{fig:OOS}}
\subfigure[OOC]
 {\includegraphics[width=2.3in,trim={0 2cm 0 0cm},clip]{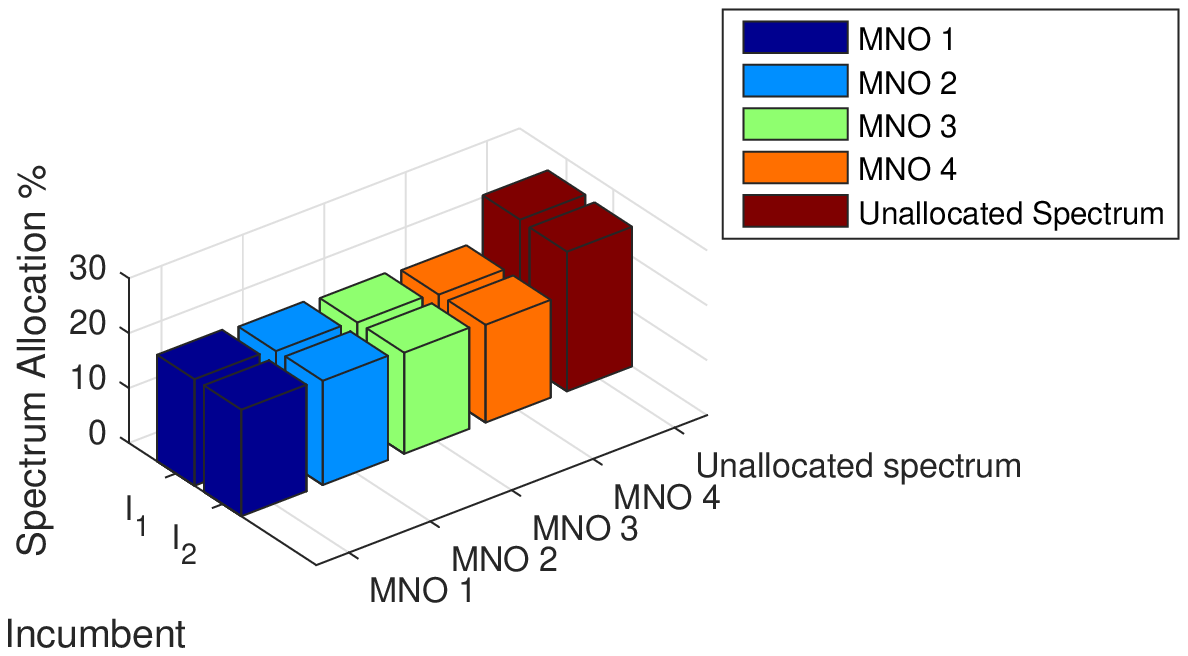}
 \label{fig:OOC}}
 \subfigure[MCS]
 {\includegraphics[width=2.3in,trim={0 2cm 0 0cm},clip]{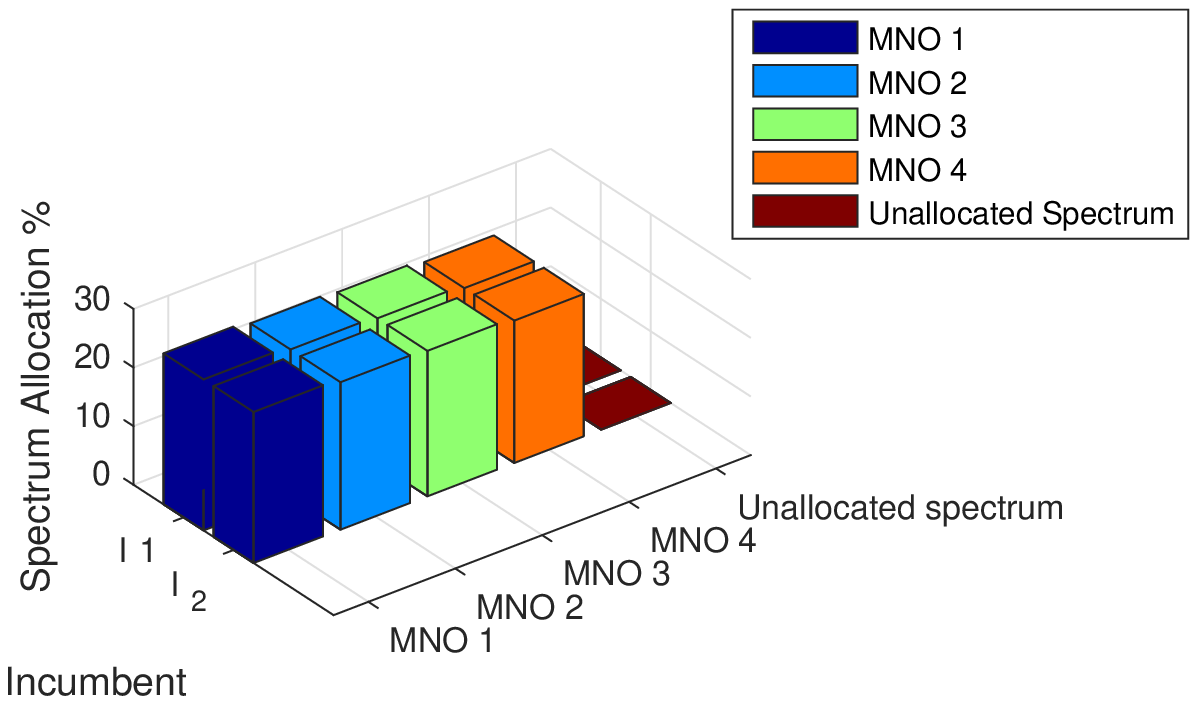}
 \label{fig:MCS}}
 \vspace{-0.2cm}
  \caption{Comparison of spectrum allocation protocols for multi-incumbent scenario. $\bar{U}_m$ for each incumbent has been plotted as well.}
  \label{fig:multi-incumbent}
  \vspace{-0.4cm}
\end{figure*}

We evaluate and compare the effect of penalty for violating the spectrum use regulations in Fig. \ref{fig:enforcement}. We plot the mean allocated spectrum as a
function of weight $\omega$. Note that an increasing value of $\omega$ implies more weight (importance) towards
fairness. As $\omega$ decreases, the weight for regulatory violation penalty increases, proportionally. We
model the parameter $PEI$ such that MNOs 1, 2, 3 and 4 have average $PEI$ values
0, 0.1, 0.2 and 0.3 respectively; the value of 0 implies no violation for MNO 1. We set $\kappa_{PEI}$ very high and assume that no MNO has value of $f(PEI)\geq \kappa_{PEI}$.

In Fig.~\ref{fig:linear_enforcemnt}, we evaluate the mean spectrum allocated to each MNO when our penalty
function is linear as stated in Section \ref{sect:enforcement}. When $\omega=1$, the available spectrum is
distributed among the MNOs equally. When $\omega$ starts decreasing, MNO 3 and MNO 4 with
large $PEI$ suffer while the other MNOs receive proportional incentive for behaving within the regulations. The
MNO 1 gains incentive monotonically as a function of decreasing $\omega$. However, MNO 2 gets incentive
initially, but is penalized when $\omega$ is very low due to increasing weight for violation
penalty and its (relatively) small $PEI$ becomes significant.
In Fig. \ref{fig:exp_enforcement}, we evaluate the effect of violation penalty when the penalty function is
a power function, i.e. $(PEI)^c$. In the numerical evaluation, we use $c=2$. In general, the larger the $c$, the slower
the penalty function growth rate in the beginning and steeper afterwards. As in Fig. \ref{fig:linear_enforcemnt}, the MNOs with large $PEI$ suffer more in terms of spectrum access as $\omega$
decreases. In contrast to linear function, the power function penalizes the MNOs at a
smaller rate initially; indeed the MNOs do not lose much share (as compared to linear function) of spectrum
when $\omega$ is relatively large. As $\omega$ decreases further, the MNOs with large $PEI$ are penalized. It
is interesting that contrary to linear function case, MNO 2 with $PEI=0.1$ is not penalized at all due to its
small $PEI$, which validates our idea behind the power penalty function that the MNOs with small violations are
not penalized much.
\subsection{Evaluation of Multi-incumbent Scenarios}
To evaluate performance for the multi-incumbent-multi-operator scenario, we assume  a scenario with $M=2$ and $N=4$ such that all of the MNOs share both $\Omega_1$ and $\Omega_2$ coalitions. Both of the incumbents offer the same amount of spectrum at each spectrum allocation instant, without loss of generality, with $B^1=B^2=100$. The MNO spectrum demand settings are the same used before for the results in Fig. \ref{fig:allocation}.

Fig. \ref{fig:multi-incumbent} compares the mean spectrum allocated to the MNOs for both incumbents 1 and 2. The spectrum allocation for all the MNOs is fair at the individual LSA network coalition level for all three proposed protocols. As the demand and offered spectrum statistics are the same for both incumbents, mean spectrum allocation statistics for both incumbents are the same as well. This verifies that the proposed protocols coordinate LSA spectrum sharing mechanism very well and maintain the individual fairness property of the individual L1 algorithm without explicitly sharing inter-coalition spectrum allocation information.
To measure the efficiency of the spectrum management algorithms, we define the mean unallocated spectrum factor $\bar{U}_m$ as the mean spectrum for an incumbent $m$ which is not allocated to any MNO, when the sum of spectrum demand from all the MNOs is greater or equal to the sum of the spectrum offered by the incumbents. Note that it is not meaningful to consider the instances where sum of the spectrum demand is less than the cumulative spectrum provided by the incumbents as unallocated spectrum may well be due to lack of demand. Thus, $\bar{U}_m$ for an incumbent $m$ is defined as,
\begin{eqnarray}
\bar{U}_m&=& \frac{1}{Q}\sum_{j=1}^Q \Big(1- \frac{\sum_{n=1}^N B_{n,m}^a(j)}{B^m(j)}\Big)
\end{eqnarray}
where $B_{n,m}^a(j)$ is the allocated spectrum to MNO $n$ from the spectrum offered by incumbent $m$ at spectrum allocation instant $j$.
$Q$ is the number of time instances when the sum of the spectrum demand from all the MNOs is equal or greater than the sum of the offered incumbent spectrum. In the numerical evaluation, as MNO minimum spectrum demand is 50, the sum of the MNO spectrum demand for 4 MNOs is never less than 200, which equals the overall spectrum offered by the incumbents at any spectrum allocation instant. Therefore, parameter $Q$ includes all the spectrum allocation instances for our spectrum demand distribution and the incumbent offered spectrum statistics.

Fig. \ref{fig:multi-incumbent} shows that the proposed L0 protocols for OOS and MCS scenarios distribute the available incumbent spectrum fully and uniformly among all the MNOs and the unallocated spectrum is zero for both incumbents. However, L0 protocol for OOC protocol is unable to allocate full incumbent spectrum, due to the one-to-one constraint on the assignment. As evident from Fig. \ref{fig:OOC}, the sum of the mean spectrum allocated to all MNOs is nearly $75\%$ and almost $25\%$ spectrum from each incumbent is wasted in spite of the demand from the MNOs. As an example, if two MNOs with demand of $50$ units of spectrum are assigned to each incumbent, the rest of the spectrum is not allocated to the MNOs regardless of the demand from the other MNOs and is a waste of resources.

To measure the efficiency of the spectrum management L0 protocols for the considered scenarios, we define another term, dissatisfaction factor $D$, for an LSA network by mean of the difference between the sum of the demand from all MNOs and the sum of the spectrum allocated at a spectrum access instant.
Let us define $D$ by,
\begin{equation}
D = \frac{1}{L}\sum_{j=1}^L\Big(1-\frac{\sum_{m=1}^M\sum_{n=1}^{N}B_{n,m}^a(j)}{\sum_{n=1}^N B_{n}^d(j)}\Big)
\label{eqn:dissatisfaction}
\end{equation}
where $L$ is the number of instances when the sum of demanded spectrum is less than the sum of spectrum offered by the incumbents. In (\ref{eqn:dissatisfaction}), only those instants are counted when the demand is less than or equal to the total spectrum offered by all the incumbents, i.e., every MNO should be able to get spectrum according to its demand in the optimal case.

For a comparison, we consider an LSA system with $M=2$ and $N=3$. All the MNOs choose spectrum demand randomly and uniformly from the vector $[50,100]$; and the spectrum made available from each incumbent is 100 as before. Fig. \ref{fig:unsatisfaction} shows the dissatisfaction factor $D$ for the proposed L0 protocols for different scenarios of the multi-incumbent case. L0 protocol for the MCS scenario has zero dissatisfaction factor, while it is the worst for OOC. L0 protocol for the MCS scenario allows MNOs to access carrier frequency bands from any number of incumbents, while it prohibit to use more than one band for one MNO for the OOC and OOS scenarios. Therefore, if the demand of an MNO is not met with spectrum from one incumbent, it is not able to fulfill the demand by accessing spectrum from the other incumbents even if the spectrum is available. This effect is more pronounced in OOC where the constraint is more tight and even incumbent is not allowed to offer residual spectrum to more than one MNO.

\section{Conclusions}
\label{sect:conclusions}
This work deals with dynamic spectrum management for LSA networks in different operating scenarios. First, we discuss a spectrum allocation algorithm for a single incumbent-multi licensee operators case which provides fair spectrum resources to all the operators. Then, we adapt the proposed algorithm for the case when the LSA licensee operators do not comply with the LSA regulations and propose appropriate penalties in terms of reduced spectrum allocation. Finally, we extend our results to multi-incumbent-multi-operators case and propose various algorithms for inter-operability of different LSA coalitions. We study the fundamental characteristics of the proposed algorithms and numerically compare the mean allocated spectrum for each licensee operator by evaluating unallocated spectrum and dissatisfaction metrics. The numerical results quantitatively show the tradeoff in terms of spectrum allocation and flexibility in spectrum access for the proposed algorithms for different LSA operation scenarios. The spectrum allocation is more efficient for the proposed L0 protocol for the MCS scenario, while it is least efficient for the OOC scenario.

\appendices
\section{Proof of Lemma 1}
\label{sect:proof_lem1}
We prove by contradiction that there exists at least a single case where multiple assignments in a single round are suboptimal. Suppose we have $N=2$ and $M=2$ and we allow multiple MNOs to accept their best offers in one round, i.e., every MNO accepts the best available spectrum offer if $\hat{A}_n>0$. Suppose MNO 1 receives offers from both incumbent 1 and 2 while MNO 2 receives offer from incumbent 2 only. Also, assume that $A_1^1>A_2^1$, $A_1^2>A_2^2$, $A_{2}^2>A_{2}^1$ and $A_1^1>A_1^2$. MNO 1 accepts incumbent 1's offer and refuses incumbent 2, while MNO 2 accepts incumbent 2's offer. We denote the pair assignments by $1\to 1$ and $2\to 2$, respectively. After the first round, MNOs 1 and 2 cannot take part further in spectrum allocation. If $B_2^d>B_{2,2}^a$, the assignment $2\to 2$ is not optimal for MNO 2 as it did not get a chance to compete for the spectrum $B^2$ (instead of allocated $B_{2,2}^a=B^2-A_1^2$) from incumbent 2 and $B^1-B_{1,1}^a$ from incumbent 1. Thus, MNO 2 was not allocated additional spectrum $\max(B^2,B^1-B_{1,1}^a)-B_{2,2}^a$ due to multiple single round assignments, which proves the lemma.

\begin{figure}
\centering
  	\includegraphics[width=3.0in]{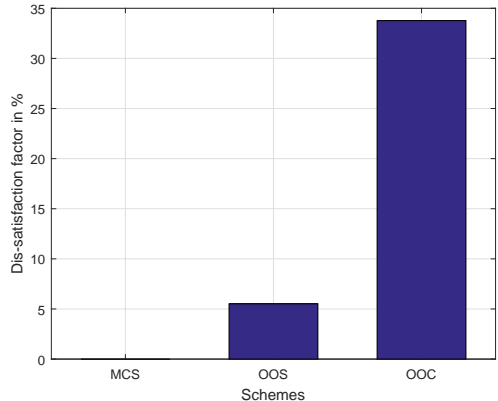}
  \vspace{-0.4cm}
   \caption{Comparison of dis-satisfaction factor for the proposed L0 protocols of various scenarios.}
	\label{fig:unsatisfaction}
\vspace{-0.4cm}
\end{figure}
\section{Proof of Lemma 2}
\label{sect:proof_lem2}
Let us denote the demand of MNO $n$ by $B^d_n$ and $B^{d_s}_n$ after a round in the cases, when MNOs are allowed to select multiple incumbents or only one incumbent per round, respectively. Note that $B_n^d(i+1)=B_n^d(i)-\sum_{m=1}^M B_{n,m}^a(i)$ for round $i$. To show that the selection of multiple incumbents in the same round allows a faster termination of the algorithm without any potential disadvantage to the MNOs, we first show that $B^d_n \leq B^{d_s}_n$. Let us start with the first round. If each MNO receives an offer only from one incumbent, then the equality holds. If at least one MNO receives an offer by two or more incumbents, its demand will obviously be lower if the MNO can accept multiple offers in the same round. Due to the fixed $PI$ of all MNOs for a spectrum allocation instant, the offers accepted during the first round will not negatively affect any of the remaining MNOs for the subsequent rounds. It should be noted that the second round (and each round after that) is exactly the same in operation as the first round, with the exception of available spectrum per incumbent and requested spectrum per MNO. Hence, for all rounds $i$, $B^d_n (i+1)\leq B^{d_s}_n(i+1)$ without loss of generality. This implies that the algorithm termination conditions $\sum_N B_n^d\to 0$ or $\sum_M B^m\to 0$ are reached at faster rate and proves Lemma \ref{lem:MCS}.
\section{Proof of Lemma 3}
\label{sect:proof_lem3}
The sum of spectrum offer to any MNO depends on its $PI$ and does not depend on single or multiple assignments of distinct pairs per round. To prove the lemma, we first show that the sum of spectrum allocated to the MNOs over rounds $j\geq i $ remains the same regardless of the single or multiple assignments $n\to m$ of various MNOs per round. If an MNO $n$ receives an offer $A_n^m (i)>0$ from an incumbent $m$ in a round $i$, it implies that it is the minimum available guaranteed spectrum for an MNO based on its $PI$ in relation to $PIs$ of the other MNOs.

In case of multiple assignments per round, if all the MNOs accept the offers $A_n^m (i)>0$ in round $i$, no MNO loses any potential spectrum offer in round $i+1$. In MCS, the residual spectrum of the incumbents is on offer again in round $j>i$ and the MNOs are allowed to participate again in round $i+1$ if their spectrum demand $B_n^d(i)$ was not met in round $i$. They compete for the residual spectrum that was not assigned in round $i$ as some of the MNOs are not interested in accepting the offers.

If only a single assignment per round is allowed, all the MNOs other than the one selected in round $i$, receive (at least) all the offers again in the next round as $PIs$ do not change. In addition, some of the MNOs may receive extra spectrum offers as the selected MNO $n^*$ might have not accepted some offers in round $i$.

In both cases, the MNO with the higher $PI$ cannot be offered spectrum as long as the MNO with the lower $PI$ does not explicitly refuse to accept the spectrum offer. Therefore, the amount of spectrum offered to an MNO is independent of the single or multiple distinct assignments $n\to m$ per round.

To prove that multiple assignment case is faster in algorithm convergence, let us denote the available spectrum  of incumbent $m$ by $B^m(i+1)$ and $\hat{B}^m(i+1)$
after round $i$ in the cases of  multiple assignments or single assignment per round, respectively. Note that $B^m(i+1)=B^m(i)-\sum_{n=1}^N B_{n,m}^a(i)$ for round $i$. In case of a single assignment per round there exists only one $(n,m)$ such that $B_{n,m}^a>0$. Hence $\hat{B}^m(i)\geq {B}^m(i)$, $\forall i$, which implies that the algorithm termination condition $\sum_M B^m \to 0$ is reached at a faster rate for multiple assignments, and proves Lemma 3.

\bibliographystyle{IEEEtran}
\bibliography{bibliography_nm}

\end{document}